\documentclass[10pt,final,twocolumn]{IEEEtran}

\usepackage{textcomp}
\DeclareSymbolFont{tildelow}{TS1}{cmr}{m}{n}
\DeclareMathSymbol{\tildelow}{0}{tildelow}{126}

\usepackage{bbm}
\usepackage{xcolor}
\usepackage{color}
\usepackage{graphicx}
\usepackage{epsfig}
\usepackage{subfigure}
\usepackage{amssymb}
\usepackage{amsbsy}
\usepackage{cite}
\usepackage{amsthm}
\usepackage{romannum}
\usepackage{algorithm}
\usepackage{algpseudocode}

\usepackage[margin=0.95in]{geometry}

\usepackage{float}
\usepackage[T1]{fontenc}
\usepackage{url}
\usepackage{ifthen}
\usepackage{cite}
\usepackage[cmex10]{amsmath} 
\usepackage{tikz}

\usepackage{verbatim}
\usetikzlibrary{graphs}

\usepackage{physics}

\newcommand{\ELL}{\mathcal{L}}
\newcommand{\ES}{\mathcal{S}}

\usepackage{mathtools}

\newtheorem{theorem}{Theorem}
\newtheorem{lemma}{Lemma}
\newtheorem{remark}{Remark}

\newtheorem{definition}{Definition}
\newtheorem{corollary}{Corollary}
\usepackage{dsfont}
\DeclareMathOperator*{\argmax}{arg\,max}

\usepackage{enumitem}
\usepackage{float}

\usepackage{booktabs}
\usepackage{caption}
\usepackage{tabularx}
\usepackage{makecell}

\makeatletter
\algrenewcommand\ALG@beginalgorithmic{\small}
\makeatother

\begin{document}

\title{A Stable Polygamy Approach to \\Spectrum Access with Channel Reuse}

\author{
  {Dan Ben Ami, Kobi Cohen \textit{(Senior Member, IEEE)} }
	\thanks{
		D. Ben Ami and K. Cohen are with the School of Electrical and Computer Engineering, Ben-Gurion University of the Negev, Beer-Sheva, Israel (e-mail:danbenam@post.bgu.ac.il; yakovsec@bgu.ac.il).} 
\thanks{This research was supported by the ISRAEL SCIENCE FOUNDATION (grant No. 2640/20), and the Israel Innovation Authority.}
\thanks{This work has been submitted to the IEEE for possible publication. Copyright may be transferred without notice, after which this version may no longer be accessible.}
	\vspace{-0.75cm}
}

\maketitle
\pagenumbering{arabic}
\begin{abstract}
\label{sec:abstract}
We introduce a new and broader formulation of the stable marriage problem (SMP), called the stable polygamy problem (SPP), where multiple individuals from a larger group $\mathcal{L}$ of $L$ individuals can be matched with a single individual from a smaller group $\mathcal{S}$ of $S$ individuals. Each individual $\ell\in\mathcal{L}$ possesses a social constraint set $\mathcal{C}_{\ell}$ that contains all other individuals in $\mathcal{L}$ with whom they cannot coexist harmoniously. We define a generalized concept of stability based on the preference and constraints of the individuals. We explore two common settings: common utility, where the utility of a match is the same for individuals from both sets, and preference ranking, where each individual has a preference ranking for every other individual in the opposite set. Our analysis is conducted using a novel graph-theoretical framework.

The classic SMP has been investigated in recent years for spectrum access in wireless communications to match cells or users to channels, where only one-to-one matching is allowed. By contrast, the new SPP formulation allows us to solve more general models with channel reuse, where multiple users may access the same channel simultaneously. We show that the classic SMP is a special case of the formulated SPP. Interestingly, we show that classic algorithms, such as propose and reject (P\&R) to achieve stability, and Hungarian method to maximize matching utilities are no longer efficient in the general polygamy setting. 

We develop efficient algorithms to solve the SPP in polynomial time, tailored for implementations in spectrum access with channel reuse. We analytically show that our algorithm always solves the SPP with common utility. While the SPP with preference ranking cannot be solved by any algorithm in all cases, we prove that our algorithm effectively solves it in specific graph structures representing strong and weak interference regimes. Simulation results demonstrate the efficiency of our algorithms across various spectrum access scenarios with channel reuse. 
\end{abstract}

\begin{IEEEkeywords}
Spectrum access, stable marriage, stable matching, interference graph, channel reuse. \vspace*{-0.2cm}
\end{IEEEkeywords}
\section{Introduction}
\label{sec:introduction}

As demand for wireless connectivity grows, advancements in technology and spectrum management are crucial for optimizing network efficiency and reliability. Effective spectrum access management, which governs the allocation and utilization of radio frequencies, is essential to mitigate interference and enhance network performance. Channel reuse is a fundamental concept in wireless communications that aims to increase the utilization of available radio frequencies while reducing interference between neighboring cells or users. In cellular networks, for example, channel reuse patterns are carefully designed to ensure that adjacent cells can use the same frequency without causing unacceptable levels of interference. This approach allows network operators to efficiently serve a large number of users within a given geographical area.

\subsection{Solving the Spectrum Access Problem via Matching Algorithms}

Leveraging matching algorithms presents a promising solution to the spectrum access challenge. These algorithms enable the efficient pairing of cells (or users or devices) with available spectrum bands (channels), optimizing resource utilization.

In the classic stable marriage (or stable matching) problem (SMP), introduced by Gale and Shapley in 1962 \cite{gale1962college}, 
two equal-sized sets, $\mathcal{N}_1$ and $\mathcal{N}_2$, each containing $N$ individuals, need to be matched one-to-one. Each $n_1\in\mathcal{N}_1$ has a preference ranking over the individuals in $\mathcal{N}_2$, and each $n_2\in\mathcal{N}_2$ has a preference ranking over the individuals in $\mathcal{N}_1$. A matching is considered unstable if there exists an individual $n_1\in\mathcal{N}_1$ who prefers another individual $n_2\in\mathcal{N}_2$ over the individual to whom $n_1$ is currently matched, and $n_2$ also prefers $n_1$ over the individual to whom $n_2$ is currently matched. Gale and Shapley developed an algorithm to solve the SMP in polynomial time, and the time is linear with the input of the algorithm (commonly known as the propose and reject (P\&R) algorithm). 

An alternative modeling of the problem is to consider a utility matrix for the SMP problem, known as a common utility setting. In this setting, each matching $(n_1,n_2)$ for $n_1\in\mathcal{N}_1$ and $n_2\in\mathcal{N}_2$, has a utility value $u_{(n_1,n_2)}$. This formulation can be seen as a value assigned by an expert (e.g., a matchmaker) to grade a matching, or as a social welfare utility (i.e., the sum of the preference rankings for each match). This approach is also known as the assignment problem. Here, $n_1\in\mathcal{N}_1$ prefers $n_2\in\mathcal{N}_2$ over $n_2'\in\mathcal{N}_2$ if $u_{(n_1,n_2)}>u_{(n_1,n_2')}$, and similarly $n_2\in\mathcal{N}_2$ prefers $n_1\in\mathcal{N}_1$ over $n_1'\in\mathcal{N}_1$ if $u_{(n_1,n_2)}>u_{(n_1',n_2)}$. 

The SMP with common utility was proposed in \cite{leshem2011multichannel} to address the stable spectrum access problem, aiming to achieve one-to-one stable matching between users and channels, where the utility of a match $(n_1, n_2)$ could represent the rate a cell $n_1$ achieves when using channel $n_2$. The authors developed a distributed CSMA-type algorithm, which achieves the P\&R solution \cite{leshem2011multichannel}. While this approach does not optimize the sum matching utility, it offers a simple distributed implementation and is known to yield strong performance in multichannel wireless networks. Variations of this problem have been studied in \cite{xiao2016enhance, avner2016multi, hao2021adaptive, liu2022spectrum, gafni2022distributed} and subsequent studies.

The optimal one-to-one matching that maximizes the sum of the matching utilities can be obtained using the Hungarian method, which solves the problem in polynomial time (with a cubic complexity in $N$) \cite{kuhn1955hungarian}. This approach requires a centralized solution with access to the full utility matrix. It was proposed in \cite{naparstek2013fully} to address the spectrum access problem, aiming to achieve one-to-one maximal sum-utility matching between users and channels. The authors developed a distributed implementation of the algorithm using the auction algorithm. Variations of this problem have been explored in \cite{naparstek2012bounds, alam2022multi} and other subsequent studies.

The SMP problem and its application to spectrum access typically involve one-to-one matching between cells and channels. However, in this paper, we focus on spectrum access with channel reuse, where multiple cells can share the same frequency without causing unacceptable interference. A common application of this setting is spectrum access optimization in spatial networks, where cells are spatially distributed, and each cell is within the interference range of a few others. The spectrum access problem with channel reuse has been studied under various settings and objectives, often focusing on Nash equilibria for spectral allocation (see, e.g., our previous work \cite{cohen2017distributed} and references therein). This paper, however, introduces a novel approach by addressing stable spectrum allocation in this context for the first time, requiring a new and broader formulation of the SMP.

\subsection{Main Results}

Below, we summarize the main contributions of our work in detail.

\noindent
\textbf{1) A novel problem formulation: The Stable Polygamy Problem (SPP):} We introduce a new and expanded formulation of the SMP, termed the Stable Polygamy Problem (SPP). Unlike the classic SMP, which allows one-to-one matching between equal-sized sets $\mathcal{N}_1$ and $\mathcal{N}_2$, the SPP permits multiple individuals from a larger set $\mathcal{L}$ of $L$ individuals to be matched with a single individual from a smaller set $\mathcal{S}$ of $S$ individuals. Each individual $\ell \in \mathcal{L}$ has a social constraint set $\mathcal{C}_{\ell}$, consisting of other individuals from $\mathcal{L}$ with whom they cannot coexist harmoniously. We define a generalized concept of stability, taking into account the preferences of individuals with respect to the other set while adhering to these social constraints. The precise definition of stable polygamy is provided in Section \ref{sec:system}.

It is important to notice the differences of our formulation compared to well-known problems in the academic literature, the \emph{Stable Roommates Problem (SRP)} \cite{irving1985efficient}, and the \emph{generalized assignment problem (GAP)} \cite{ozbakir2010bees, cohen2006efficient}. In SRP and its variations, the goal is to find a stable matching within an even-sized set by pairing elements into disjoint pairs (i.e., roommates). A matching is considered stable if no two unpaired elements prefer each other over their assigned roommates. This differs from SMP and SPP, as SRP allows matches between any two elements within the set, rather than between distinct classes of individuals as in $\mathcal{L}$ and $\mathcal{S}$.

In GAP and its variations, there are two distinct classes similar to SPP: a smaller group of agents (analogous to $\mathcal{S}$) and a larger group of tasks (analogous to $\mathcal{L}$). In this problem, any agent can be assigned to perform any task, with varying costs and profits depending on the assignment. GAP allows many-to-one assignments like SPP, but the constraints differ. In GAP, each agent has a budget that limits the total cost of tasks they can be assigned. The goal is to find an assignment where all agents stay within their budgets while maximizing total profit. In contrast, SPP does not impose a budget constraint on the number of individuals from $\mathcal{L}$ that can be assigned to $\mathcal{S}$. Instead, the constraints are based on social limitations that restrict which individuals from $\mathcal{L}$ can coexist within the same allocation to an entity in $\mathcal{S}$. These can be seen as social constraints ensuring harmonious coexistence within the stable polygamy framework.

\noindent
\textbf{2) A new graph-theoretical analysis:} We introduce a novel graph-theoretical interpretation of the SPP using a vertex coloring-based formulation. Traditionally, matching problems are modeled using a weighted bipartite graph, where the sets $\mathcal{L}$ and $\mathcal{S}$ are represented by vertices, divided into two disjoint and independent sets. In contrast, our approach employs a vertex coloring-based formulation. Here, only the set $\mathcal{L}$ is represented by the vertices (i.e., cells or users) in the graph, while the set $\mathcal{S}$ is represented by the colors (channels) used to color these vertices. Social constraints are represented by the edges, preventing vertices from being matched with the same color. A common application for this setting is spectrum access optimization in spatial networks, where cells are spatially distributed, and each cell is in the interference range of a few other cells (see \cite{hou2013proportionally, cohen2017distributed} and subsequent studies). Here, each vertex-color pairing is evaluated based on utility or preference ranking between the vertex and color sets. A stable polygamy is thereby represented through a new formulation of stable vertex coloring, as introduced in this paper. Through this graph-theoretical analysis, we obtain that the classic SMP is a special case of the formulated SPP, where the graph is complete, and the sizes of $\mathcal{L}$ and $\mathcal{S}$ are equal.

\noindent
\textbf{3) Solving SPP with common utility for spectrum access with channel reuse:} As previously mentioned, the special case of SMP with common utility is also referred to as an assignment problem. In this context, the optimal one-to-one matching that maximizes the total matching utility can be found using the Hungarian method, which solves the problem centrally in polynomial time \cite{kuhn1955hungarian}. However, using our graph-theoretical analysis, we show that maximizing the total matching utility in the polygamy setting considered here is NP-hard. Consequently, the Hungarian method is efficient only in the specific case of SMP and cannot be generalized to the polygamy setting. Beyond the strong performance that stable matching is known to achieve in multichannel wireless networks, this result further highlights the rationale for applying SPP to spectrum access with channel reuse, as explored in this paper.

We develop a distributed stable spectrum access algorithm with channel reuse by solving the SPP, referred to as Distributed SPP for Spectrum Access and Reuse (DSSAR). DSSAR can be implemented using an opportunistic CSMA mechanism \cite{zhao2005opportunistic} or through a local control messages mechanism \cite{cohen2017distributed}. We analytically prove that DSSAR solves the SPP in polynomial time relative to the size of $\mathcal{L}$.

\noindent
\textbf{4) Solving SPP with preference ranking for spectrum access with channel reuse:} As previously mentioned, under preference ranking, Gale and Shapley developed the P\&R algorithm to solve the SMP in polynomial time \cite{gale1962college}. However, using our graph-theoretical analysis, we show that a solution to SPP under preference ranking does not always exist. Particularly, the P\&R algorithm highly depends on the special structure of the SMP and no longer efficient in solving the SPP. 

We introduce a new algorithm, termed Re-Propose and Reject (RP\&R), which allows for re-proposing matches. This algorithm can be implemented through a centralized mechanism with low complexity, making it well-suited for integration into 5G and beyond technologies with centralized node deployments. We prove analytically that RP\&R solves the SPP in specific graph structures that represent strong and weak interference regimes within wireless communication networks. For more general graph structures, we conducted extensive simulations using acyclic social constraint graphs. In all cases, the algorithm successfully achieved stable polygamy within $L$ iterations. 
%

\subsection{Other Related Work}

Another set of related work on multi-user channel allocation has approached it from the angle of game theoretic and congestion control (~\cite{han2005fair, menache2008rate, candogan2009competitive, menache2011network, law2012price, cohen2013game, wu2013fasa, singh2016combined, cohen2016distributedToN, cohen2017distributed, malachi2020queue} and references therein), hidden channel states\cite{yemini2019restless}, and classic graph coloring (\cite{wang2005list, wang2009improved,  checco2013learning, checco2014fast} and references therein). Here, the problem is fundamentally different from few aspects. First, traditional coloring might not be feasible here, since the number of channels might be small, and not all vertices might assigned with colors. Second, the objective of stable coloring in terms of preferences is a new concept. Game theoretic aspects of the problem have been investigated from both non-cooperative (i.e., each user aims at maximizing an individual utility) \cite{menache2008rate, candogan2009competitive, singh2016combined, cohen2016distributedToN, cao2018distributed}, and cooperative (i.e., each user aims at maximizing a system-wide global utility) \cite{han2005fair, leshem2006bargaining, cohen2017distributed, bistritz2018approximate} settings. 

Another set of related work consider the setting where the environment (e.g., the utility matrix) is unknown and need to be learned. These studies have been focused on one-to-one mapping. In \cite{bistritz2018distributed, bistritz2021game}, the authors developed a distributed algorithm to maximize network utility, addressing the one-to-one mapping problem in a distributed manner. In \cite{avner2016multi}, the SMP was addressed via multi-armed bandit (MAB) settings, modeling the channels as arms that need to be learned. A number of studies have developed distributed learning algorithms for restless Markovian channel model, where each channel yields the same expected rate for all users \cite{Tekin_2012_Online, liu2012learning, gafni2020learning}. In our recent work, a more general model was analyzed where each channel yields a different expected rate for each user \cite{gafni2022distributed}, and a model that includes exogenous global Markov process that controls the states of arms \cite{gafni2022learning}. These studies use the framework of restless MAB (RMAB). Model-free learning strategies were developed in \cite{naparstek2018deep,liu2021dynamic, bokobza2023deep, paul2023multi, cohen2024sinr}. Finally, none of these studies have considered the problem of achieving provable stable strategies in the polygamy setting considered here. 

Consistent with subsequent studies that have integrated learning strategies into spectrum access models, including classic SMP and game-theoretic approaches as discussed above, we expect that this work will inspire further exploration of learning strategies within the newly formulated SPP introduced in this paper. Consequently, this new approach is anticipated to make a significant contribution to the field of spectrum access.

\section{System Model and Problem Formulation}
\label{sec:system}

We begin by describing the polygamy setting considered in this paper, followed by an explanation of the applications of SPP in spectrum access with channel reuse.

\subsection{Formulation of the Stable Polygamy Problem (SPP)}

Let $\mathcal{L}$ be a larger set comprising $L$ individuals, and let $\mathcal{S}$ be a smaller set consisting of $S$ individuals, plus a virtual individual $V$ who holds the lowest preference compared to all real individuals $s \in \mathcal{S}$. We will explain the role of the virtual individual later. Individuals from set $\mathcal{L}$ can only be matched to individuals from set $\mathcal{S}$. Intra-set matches are not permitted. Each individual $\ell\in\mathcal{L}$ possesses a social constraint set $\mathcal{C}_{\ell}$ that contains all other individuals in $\mathcal{L}$ with whom they cannot coexist harmoniously with the same individual $s\in\mathcal{S}\setminus\{V\}$. 

We note that it is permissible to have polygamy scenarios where certain individuals (say $\ell$) in $\mathcal{L}$ are not matched with any real individual in $\mathcal{S}\setminus\{V\}$ due to social constraints. For the sake of mathematical correctness, we consider such cases by treating $\ell$ as if they are matched with the virtual individual $V \in \mathcal{S}$. This mapping is allowed for all $\ell \in \mathcal{L}$ where the constraints are relaxed and not considered.

To streamline the terminology related to social constraints among individuals, we introduce the following definitions.

\begin{definition}[Social compatibility]
    Individuals $\ell_1\in\mathcal{L}$ and $\ell_2\in\mathcal{L}$ are considered socially compatible if $\ell_2 \notin \mathcal{C}_{\ell_1}$ and $\ell_1 \notin \mathcal{C}_{\ell_2}$. Otherwise, they are termed socially incompatible.
\end{definition}

\begin{definition}[Social availability]
An individual $s\in\mathcal{S}\setminus\{V\}$ is considered socially available to $\ell\in\mathcal{L}$ if all individuals matched with $s$ who are preferred by $s$ over $\ell$ are socially compatible with $\ell$.
\end{definition}

It is desired to match multiple individuals in $\mathcal{L}$ to a single individual in $\mathcal{S}$ given the social constraints described above, and preference values. We examine two prevalent settings for preference values commonly explored in matching problem formulations: \emph{preference ranking}, and \emph{common utility}.\vspace{0.2cm} 

\noindent
\textbf{Preference ranking:} In this setting, each individual maintains a preference ranking for every other individual in the opposing set. The preference matrix for individuals in $\ELL$ is denoted by $R^\ELL \in \{1, 2, 3, \ldots, S \}^{L\cross S}$, where the $(i, j)$-th entry signifies the ranking assigned to $s_j\in\ES$ by $\ell_i\in\ELL$, where rank $1$ stands for the most preferred and rank $S$ the least preferred. Specifically, each row in $R^\ELL$ constitutes a permutation of the set $\{1, 2, 3, \ldots, S\}$. The virtual individual is always assigned the lowest preference value, $S$, meaning that all entries in the last column of $R^\mathcal{L}$, associated with the virtual individual, are equal to $S$. Analogously, the preference matrix for individuals in $\ES$ is denoted by $R^\ES \in \{1, 2, 3, \ldots, L\}^{L\cross S}$, where the $(i, j)$-th entry signifies the ranking assigned to $\ell_i\in\ELL$ by $s_j\in\ES$, where rank $1$ stands for the most preferred and rank $L$ the least preferred. Specifically, each column in $R^\ES$ constitutes a permutation of the set $\{1, 2, 3, \ldots, L\}$. For the virtual individual, we set the rankings arbitrarily from $1$ to $L$, meaning the last column of $R^\mathcal{S}$, associated with the virtual individual, is equal to $(1, 2, ..., L)^T$.
\vspace{0.2cm}

\noindent
\textbf{Common utility:} Let $U$ be a common utility matrix with dimension $L\cross S$ for the SPP problem. In this setting, each match $(\ell,s)$ for $\ell\in\ELL$ and $s\in\ES\setminus\{V\}$, has a utility value $u_{(\ell,s)}> 0$. A match between $\ell\in\ELL$ and the virtual individual $V$ has utility zero: $u_{(\ell,V)}= 0$. We can view this formulation as a value that an expert (e.g., match-maker) grades a matching, or a social-welfare utility (i.e.,  $U=R^\ELL+R^\ES$). This is also known as an assignment problem. Here, $\ell\in\ELL$ prefers $s_1\in\ES$ over $s_2\in\ES$ if $u_{(\ell,s_1)}>u_{(\ell,s_2)}$, and similarly $s\in\ES$ prefers $\ell_1\in\ELL$ over $\ell_2\in\ELL$ if $u_{(\ell_1,s)}>u_{(\ell_2,s)}$. This setting was considered in the special case of SMP to solve the stable spectrum access problem in \cite{leshem2011multichannel} (and subsequent studies) to obtain one-to-one stable matching between users and channels.

Below, we define admissible polygamy, harmonious polygamy, and stable polygamy. 

\noindent
\begin{definition}[Admissible polygamy]
An admissible polygamy $\phi: \mathcal{L} \rightarrow \mathcal{S}$ is a matching of individuals in $\mathcal{L}$ to individuals in $\mathcal{S}$, where each individual in $\mathcal{L}$ is matched to exactly one individual in $\mathcal{S}$, and multiple individuals in $\mathcal{L}$ can be matched to the same individual in $\mathcal{S}$.
\end{definition}

\noindent
\begin{definition}[Harmonious polygamy]
    A harmonious polygamy is an addmisible polygamy $\phi: \mathcal{L} \rightarrow \mathcal{S}$ in which there are no individuals $\ell\in \mathcal{L}$, $\ell'\in\mathcal{C}_{\ell}$ which are matched to the same individual $s\in\mathcal{S}\setminus\{V\}$.
\end{definition}

\noindent
\begin{definition}[Stable polygamy]
    A harmonious polygamy $\phi: \mathcal{L} \rightarrow \mathcal{S}$ is considered stable if for every individual $\ell_1\in L$ that prefers another individual $s\in\mathcal{S}$ over the individual to whom it is already matched, then there exists individual $\ell_2$ which is matched to $s$, which is preferred by $s$ over $\ell_1$, such that $\ell_2\in \mathcal{C}_{\ell_1}$ or $\ell_1\in \mathcal{C}_{\ell_2}$.
\end{definition}

We now define the stable polygamy problem (SPP):\vspace{0.2cm}

\noindent
\begin{definition}[Stable polygamy problem] The objective of the SPP is to find a stable polygamy $\phi: \mathcal{L} \rightarrow \mathcal{S}$ for the sets of individuals. 
\end{definition}

\subsection{Applications in Spectrum Access with Channel Reuse}
\label{ssec:applications}

We will now explain the applications of SPP in spectrum access with channel reuse. Channel reuse is a fundamental concept in wireless communications aimed at increasing the utilization of available radio frequencies while reducing interference between neighboring cells or users. In interference environments, it is common to model mutual interference by links between cells or users, allowing those close to each other (i.e., connected by interference links) to use different channels to ensure frequency separation. Users located far from each other (i.e., not connected by interference links) can use the same channel. In this context, the set of vertices represents the large set of users $\mathcal{L}$, the interference links represent the social constraints, and the set $\mathcal{S}$ represents the set of channels. This approach enables network operators to efficiently serve a large number of users within a given geographical area.

Furthermore, channel reuse can be employed in high-interference environments, where users are classified by priority. For instance, a high-priority user can be assigned to the same channel as low-priority users, and by adjusting their transmission parameters accordingly, they will receive the desired QoS. However, two or more high-priority users cannot be matched to the same channel. These restrictions define the social constraints.

Note that both preference ranking and common utility are applicable for modeling utility in spectrum access applications. In the common utility model, the achievable rate (or a function of the rate) of user $\ell$ over channel $k$ can be represented by the common utility associated with the match between $\ell \in \mathcal{L}$ and $k \in \mathcal{S}$. In the preference ranking model, each user $\ell \in \mathcal{L}$ can rank the channels they wish to access based on the rate they achieve on each channel. Meanwhile, a control manager ranks the preferred users on each channel based on factors such as user characteristics and the potential interference a user might cause to external usage of the channel or channel operations.

\section{Graph-Theoretic Analysis of SPP and Its Applications in Spectrum Access with Channel Reuse}
\label{sec:graph}

In this section, we analyze the SPP using the framework of graph theory. We will analytically demonstrate that classic algorithms designed to solve the special case, the SMP, are no longer efficient when applied to the generalized setting, necessitating new algorithmic developments to address the SPP. We will then provide specific applications of the graph-theoretic framework in spectrum access with channel reuse.

\subsection{Formulation of SPP as a Generalized Vertex Coloring}
\label{subsec:Reformulation}

Consider a graph where each vertex represents an individual from set $\mathcal{L}$. Each individual $\ell$, represented as a vertex, is connected by a directed edge to all other vertices corresponding to individuals in $\mathcal{L}$ that are within its social constraints $C_\ell$. For harmonious polygamy, no two individuals of $\mathcal{L}$ can be paired with the same individual $s\in\mathcal{S}\setminus\{V\}$, even if only one of them falls within the social constraints of the other. Therefore, the graph can be characterized by its underlying undirected graph, which is formed by removing edge directionality and eliminating duplicate edges. From this point forward, we will refer to this graph as an undirected graph.

Each individual in $\mathcal{S}$ will be represented by a distinct color. A valid generalized vertex coloring ensures that no adjacent vertices share the same color. In the preference ranking setting, each vertex has a preference rating for each color in set $\mathcal{S}$, represented by the matrix $R^\mathcal{L}$, while each color in $\mathcal{S}$ has a preference rating for each vertex, denoted by $R^\mathcal{S}$. In the common utility setting, each coloring of $\ell \in \mathcal{L}$ by $s \in \mathcal{S}$ has an associated utility value $u_{(\ell,s)} > 0$, given by the matrix $U$. It is important to note that within this generalized coloring problem, some nodes may remain uncolored. These uncolored nodes correspond to individuals in $\mathcal{L}$ who were unmatched in the SPP problem. To address this, we define a virtual color $V$, which can be assigned to neighboring nodes without conflict, in line with the concept of the virtual individual in the SPP.

In this formulation, the definitions of admissible, harmonious, and stable polygamy are adapted as follows: $\mathcal{L}$ represents the set of nodes, $\mathcal{S}$ denotes the set of colors, social constraints $\mathcal{C}_{\ell}$ are determined by the edges, and preference ranking and common utility define the coloring payoff. In this context, stable polygamy corresponds to a valid generalized $\phi$-coloring, in which for each vertex $\ell$ and for each color $s$ that is more preferred by $\ell$ then its assigned color $\phi(\ell)$, there exists at least one adjacent vertex colored with $s$ that $s$ prefers over $\ell$.

\begin{corollary}[SMP as a special case of SPP]
Consider a complete graph in the generalized vertex coloring formulation of SPP with vertex set $\mathcal{L}$ and color set $\mathcal{S}$, under the condition $|\mathcal{L}| = |\mathcal{S} \setminus\{V\}|$. In this scenario, the SPP reduces to the SMP.
\end{corollary}

\begin{proof}
The condition $|\mathcal{L}| = |\mathcal{S} \setminus\{V\}|$ ensures that the two groups are equally sized in terms of real individuals, similar to the SMP. Since the graph is complete, the social constraints enforce a one-to-one mapping between the sets $\mathcal{L}$ and $\mathcal{S}$, meaning no individual $s \in \mathcal{S}$ will be matched with more than one individual. As the virtual individual $V$ is the least preferred by all $\ell \in \mathcal{L}$, it becomes irrelevant and will not be assigned to any user in this specific setting.
\end{proof}

Note that our generalized vertex coloring formulation offers a new interpretation of the classic matching problem in the monogamy setting, as in SMP. In this context, the matching does not occur between two sets of nodes in a bipartite graph. Instead, in the generalized vertex coloring formulation, the matching is between vertices and colors. This generalized formulation later helps us demonstrate the computational intractability of achieving an optimal solution in the polygamy setting compared to the monogamy setting.

\subsection{Computational Complexity of Finding the Optimal Solution under Common Utility}
\label{ssec:computational}
In the classic monogamy setting of standard one-to-one matching with common utility (also known as the assignment problem), an optimal solution is defined as the solution that maximizes the sum utility over all matches, $\max \sum_{\ell_i \in \mathcal{L}} u_{\ell_i,\phi(\ell_i)}$. To achieve an optimal outcome in this setting, the Hungarian method can be applied, yielding an optimal solution with a computational complexity of $O(\mathcal{L}^3)$. However, this paper extends to a polygamy setting of many-to-one matching that incorporates social constraints among individuals within the larger group. Therefore, we present the following theorem:
\begin{theorem}
\label{NP_hard}
    Determining an optimal solution that maximizes the sum utility over all matches in the polygamous matching setting with social constraints is NP-hard.
\end{theorem}
\begin{proof}
    Assume, for the sake of contradiction, that a polynomial-time algorithm with $L$ exists for finding an optimal solution to the problem. Consider the scenario where the utility matrix $U$ has the same value for all entries. In such a case, an optimal solution would aim to maximize the number of nodes that receive a color assignment. Specifically, the objective would be to color as many nodes as possible without violating any social constraints. Given that the chromatic number is bounded by the total number of nodes $\ELL$, one could iteratively apply the aforementioned algorithm, incrementing the available color count by one in each iteration, until all nodes are successfully colored. Since the previous iteration failed to color all nodes, the chromatic number—defined as the minimum number of colors needed for a valid coloring—must correspond to the color count in the current iteration. Employing this approach, we can determine the chromatic number using at most $\ELL$ iterations of the polynomial-time algorithm. This contradicts the established NP-hard nature of computing the chromatic number of the graph.
\end{proof} 

\begin{remark}
\textnormal{Theorem \ref{NP_hard} further underscores the motivation to apply SPP in spectrum access with channel reuse, as considered in this paper. While in the classic monogamy setting of one-to-one matching, the Hungarian method provides an efficient polynomial-time solution to maximize the sum utility (though with higher complexity than SMP), in the generalized polygamy setting of many-to-one matching, an efficient solution that maximizes the sum utility cannot be obtained. In Section \ref{sec:solving_utility}, we will show that a solution to the SPP under common utility always exists, and we will develop an efficient distributed algorithm to solve the SPP in this setting.} 
\end{remark}

\subsection{Unsolvable Scenarios in Preference-Based SPP}

In contrast to SPP under common utility, we show below that a solution to SPP under preference ranking does not always exist. Nevertheless, in Section \ref{sec:solving_preference}, we will develop an efficient centralized algorithm to solve the SPP, demonstrate its effectiveness through simulations across a wide range of problems, and prove its correctness in certain special cases.

\begin{theorem}
\label{th:Not}
Not all instances of SPP with preference ranking are solvable.    
\end{theorem}
\begin{proof}
We prove the theorem by providing an example of an unsolvable instance of the problem. Consider a larger set of size $|\mathcal{L}|=5$ and a smaller set of size $|\mathcal{S}\setminus\{V\}|=2$ with the following preference matrices:
$$
R^\mathcal{L} = 
\begin{pmatrix}
1 & 2 & 3\\
2 & 1 & 3\\
1 & 2 & 3\\
2 & 1 & 3\\
1 & 2 & 3\;
\end{pmatrix},\;
R^\mathcal{S} = 
\begin{pmatrix}
2 & 5 & 1\\
4 & 2 & 2\\
3 & 1 & 3\\
1 & 4 & 4\\
5 & 3 & 5
\end{pmatrix},
$$
where the third column is associated with the virtual individual $V\in\mathcal{S}$. Let the social constraints be given by the following graph:

\begin{figure}[H]
  \centering
  \includegraphics[width=0.2\textwidth]{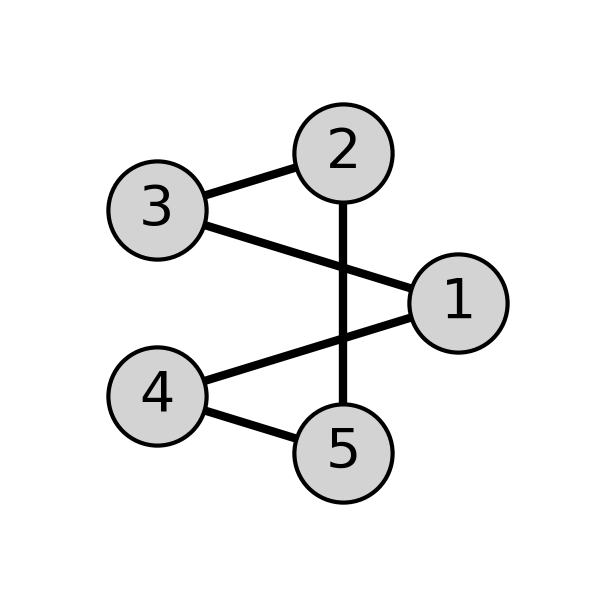}
  \label{fig:example_graph}
\end{figure}
\vspace{-0.5cm} 

\noindent
As explained in Section \ref{subsec:Reformulation}, this graph structure indicates that an edge between two nodes (i.e., individuals in $\mathcal{L}$) signifies social incompatibility between them.

By examining all possible polygamy matches for this example, it can be verified that a stable polygamy does not exist.
\end{proof}

\subsection{Applications of the SPP Graph-Theoretic Framework in Spectrum Access with Channel Reuse}

Here, we formalize the widely adopted interference graph model for spectrum access in communication networks with channel reuse \cite{srikant2013communication, cohen2017distributed}, as discussed in Section \ref{ssec:applications}. For the first time, we introduce a stable matching formulation for this model through the generalized SPP framework. In subsequent sections, we will develop efficient algorithms to solve the problem.

Consider a wireless network comprising a set $\mathcal{L} = \{1, 2, ..., L\}$ of cells (i.e., the larger set) and a set $\mathcal{S} = \{1, 2, ..., S\}$ of shared channels (i.e., the smaller set). In alignment with our formulation, channel $S$ is defined as a zero-utility channel representing the virtual individual $V$ in the smaller set. 

We focus on a spatial wireless network, where each cell is in the interference range of a few (but not necessarily all) other cells. We assume symmetric interference ranges for all cells in the sense that cell $\ell_1$ is in cell $\ell_2$'s interference range only if cell $\ell_2$ is in cell $\ell_1$'s interference range for all $\ell_1, \ell_2\in\ELL$. We refer to cells in the same interference range as \emph{neighbors}, and represent by $\mathcal{C}_\ell\subseteq \ELL\setminus\{\ell\}$ the set of cell $\ell$'s neighbors. Each cell is allowed to choose a single channel for transmission.

The network topology can be represented by an undirected graph $G= (\ELL, \mathcal{E})$, where the set of cells are represented by the vertices and the interference relationships between cells are represented by the set of edges $\mathcal{E}$. An edge $(\ell_1, \ell_2)\in\mathcal{E}$ means that cells $\ell_1$ and $\ell_2$ are in the same interference range. The set of cell $\ell$'s neighbors $\mathcal{C}_{\ell}$ is represented by vertices directly connected to vertex $\ell$ excluding vertex $\ell$ itself. It is desired to allocate different channels to cells in the same interference range to avoid interference between cells. 

In the common utility model, the achievable rate (or a function of the rate) of cell $\ell$ over channel $k$ can be represented by the common utility associated with the match between cell $\ell \in \mathcal{L}$ and channel $k \in \mathcal{S}$. In the preference ranking model, each cell $\ell \in \mathcal{L}$ can rank the channels they wish to access based on the rate (or a function of the rate) they achieve on each channel. Meanwhile, a control manager ranks the preferred users on each channel based on factors such as user characteristics and the potential interference a user might cause to external usage of the channel or channel operations.

\section{Solving SPP with Common Utility for Spectrum Access with Channel Reuse}
\label{sec:solving_utility}

In this section, we develop an efficient distributed algorithm to solve the SPP with common utility in polynomial time, enabling stable spectrum access with channel reuse through simple distributed mechanisms.

\subsection{The Distributed SPP for Spectrum Access and Reuse (DSSAR) Algorithm}

As discussed in Section \ref{ssec:computational}, maximizing the sum matching utility is NP-hard in the polygamy setting, highlighting the motivation to apply SPP in spectrum access with channel reuse, as considered in this paper. Therefore, in this section, we develop a distributed algorithm, dubbed Distributed SPP for Spectrum Access and Reuse (DSSAR), to solve the stable spectrum access with channel reuse by addressing the SPP.

We describe the flow of DSSAR, prove that it solves the SPP in polynomial time, and explain its implementation in a distributed manner.

The pseudocode of DSSAR is provided in Algorithm \ref{alg:DSSAR}. DSSAR takes as input a larger set of individuals $\mathcal{L}$ (cells), a smaller set of individuals $\mathcal{S}$ (channels), a utility matrix $U$ of size $L\times S$ (where the entries in the last column correspond to the virtual channel and are thus zero, $u_{(\ell,S)}=0$ for $\ell=1, ..., L$), and social constraints $C_{\ell}$ for each $\ell\in\mathcal{L}$.

We initialize DSSAR by defining a matrix $\tilde{U}$ and assigning $U$ to it (the role of $\tilde{U}$ will be explained later). The initial matching assigns the virtual channel $S$ to each cell. 

In the routine, DSSAR runs for at most $L$ iterations, with each iteration assigning a channel to one cell. The selection of the next cell to be assigned a channel is based on the utility value. Therefore, we use the definition of $\tilde{U}$ to be a zero-updated utility matrix, depending on the updated assignment. Specifically, at each iteration, once cell $\ell^*$ is assigned a channel $s^*$, we update each entry $\tilde{u}_{(\ell^*, s)}$ in $\tilde{U}$ to be zero for all $s=1, ..., S-1$. This indicates that $\ell^*$ has been assigned and will not be selected again for assignment. Additionally, we update each entry $\tilde{u}_{(\ell, s^*)}$ to be zero for all $\ell\in\mathcal{C}_{\ell^*}$ (or $\ell^*\in\mathcal{C}_{\ell}$ if non-symmetric interference range is allowed). This means that cells in the interference range of $\ell$ will not be assigned to channel $s^*$. If, at any step, all entries in $\tilde{U}$ are zero, the routine can terminate early, exiting the loop (line 7). \vspace{0.2cm}

\begin{algorithm}
\caption{The DSSAR Algorithm}
\begin{algorithmic}[1]
\Require $\ELL, \ES, U, \mathcal{C}_{\ell} \; \forall \ell \in \ELL$
\State \textbf{Initialization:} 
$\tilde{U} \gets U$, $\phi(\ell) \gets V\;\forall\ell\in\mathcal{L}$  
\For{$iteration = 1, 2, \ldots, L$}
    \State $(\ell^*, s^*) \gets \argmax_{\ell,s} \tilde{u}_{(\ell,s)}$
    \If{$\tilde{u}_{(\ell^*,s^*)} > 0$}
        \State $\phi(\ell^*) \gets s^*$
    \Else
        \State \textbf{break}
    \EndIf
    \State $\tilde{u}_{(\ell^*, s)} \gets 0$ \textbf{for all} $s = 1, \ldots, S-1$
    \State $\tilde{u}_{(\ell, s^*)} \gets 0$ \textbf{for all} $\ell$ \textbf{such that} $\ell \in \mathcal{C}_{\ell^*}$ \textbf{or} $\ell^* \in \mathcal{C}_{\ell}$
\EndFor
\end{algorithmic}
\label{alg:DSSAR}
\end{algorithm}

We next prove that DSSAR solves the SPP.  

\begin{theorem}
The DSSAR algorithm solves the SPP, and the running time complexity is $O(L^2 S)$\vspace{0.2cm} 
\end{theorem}

\begin{proof}
First, note that the solution of DSSAR is admissible and harmonious. This can be verified as follows. Once $\ell^*\in\ELL$ is assigned a channel $s^*\in\ES$, line 9 updates: $\tilde{u}_{(\ell^*, s)}\gets 0$ for all $s=1, ..., S-1$. Consequently, in the following iterations $\ell^*$ will not be assigned a channel. As a result, each cell is assigned a channel at most once, ensuring the solution is admissible. Also, line 10 updates: $\tilde{u}_{(\ell, s^*)}\gets 0$ for all $\ell$ such that $\ell\in\mathcal{C}_{\ell^*}$ or $\ell^*\in\mathcal{C}_{\ell}$. Consequently, in the following iterations all cells $\ell$ such that $\ell\in\mathcal{C}_{\ell^*}$ or $\ell^*\in\mathcal{C}_{\ell}$, will not be assigned to channel $s^*$. As a result, the solution is harmonious. 
Finally, we prove stability by contradiction. Assume there exists cell $\ell_1\in\ELL$ that prefers another channel $s\neq\phi(\ell_1)$ over the channel it is currently matched to, i.e., $u_{(\ell_1, s)}>u_{(\ell_1, \phi(\ell_1))}$, and assume by contradiction that there is no other cell $\ell_2$, such that $\ell_2\in \mathcal{C}_{\ell_1}$ or $\ell_1\in \mathcal{C}_{\ell_2}$, which is matched to $s$, i.e., $\phi(\ell_2)=s$, and fulfills: $u_{(\ell_2, s)}>u_{(\ell_1, s)}$. By the construction of the algorithm, since there is no conflict with any other assignment of channel $s$, cell $\ell_1$ would have attempted to acquire channel $s$ before accepting its current assignment $\phi(\ell_1)$. Then, $\phi(\ell_1)\gets s$, contradicting the assumption. Hence, the solution is stable. 

Next, observe that the algorithm iterates $L$ times, with each iteration requiring $O(LS)$ operations (for finding the maximum and making assignments). Consequently, the overall time complexity of the algorithm is $O(L^2 S)$.
\end{proof}

\subsection{Distributed Mechanisms for Implementing DSSAR}

Finally, we outline simple distributed mechanisms for implementing DSSAR. This can be achieved through opportunistic CSMA in the interference graph \cite{zhao2005opportunistic}, or by enabling local control messages between neighboring nodes, as proposed in distributed spectrum access algorithms \cite{cohen2017distributed}. \vspace{0.2cm} 

\subsubsection{Implementation of DSSAR via opportunistic CSMA}

CSMA is a widely-used communication protocol for managing spectrum sharing among multiple users. It allows devices within close range to monitor the channel for existing transmissions before initiating their own. If the channel is detected as idle, the device can proceed with its transmission. However, if ongoing activity is detected, the device will defer its transmission and may seek an alternative channel in multi-channel systems.

To implement DSSAR in a distributed manner, we can use opportunistic CSMA \cite{zhao2005opportunistic, cohen2010time, leshem2011multichannel}. In this approach, a backoff function maps utility values to backoff times, where the backoff time decreases monotonically with higher utility values. Consequently, the cell with the highest utility for a given channel will have the shortest wait time before transmission. Cells within the interference range of this cell will detect that the channel is occupied and will refrain from transmitting on that channel. This process, which aligns with line 10 of the algorithm, assumes symmetric interference ranges: if cell $\ell_1$ is within cell $\ell_2$'s interference range, then cell $\ell_2$ is also within cell $\ell_1$'s interference range. The procedure continues until all $L$ cells complete line 5, excluding those left unassigned due to social constraints. This approach provides a simple distributed implementation of DSSAR.\vspace{0.2cm}

\subsubsection{Implementation of DSSAR via local control messages}
Alternatively, DSSAR can be implemented by enabling local control messages between neighboring cells in the graph, as suggested in distributed spectrum access algorithms \cite{cohen2017distributed}. In this approach, a decreasing backoff function based on utility values is still employed. When a cell $\ell$ selects a channel $s$, it sends a control message to its neighboring cells to inform them that channel $s$ is in use. This method ensures a reliable implementation if perfect carrier sensing is not possible, but it requires the exchange of local control messages.

\section{Solving SPP with Preference Ranking for Spectrum Access with Channel Reuse}
\label{sec:solving_preference}

In this section, we address solving spectrum access with channel reuse using SPP with preference ranking. We develop a Re-Propose and Reject (RP\&R) algorithm to establish stable polygamy in a preference ranking scenario. We present the RP\&R algorithm to solve the problem, assuming a solution to SPP exists (note that a solution does not always exist, as shown in Theorem \ref{th:Not}). While we analytically prove that RP\&R solves SPP in specific interesting cases, simulation results demonstrate its strong performance across a wide range of general scenarios.

\subsection{The Re-Propose and Reject (RP\&R) Algorithm}

The pseudocode of RP\&R is provided in Algorithm \ref{alg:rpr}. For presenting the algorithm, we denote the preference sequence of $s$ as $\ell^s_1, \ell^s_2, \ldots, \ell^s_L$, where $\ell^s_1$ signifies the top preference of $s$, and $\ell^s_L$ denotes its least favored option.

RP\&R takes as input a larger set of individuals $\mathcal{L}$ (cells), a smaller set of individuals $\mathcal{S}$ (channels), a preference ranking matrix $R^\mathcal{L}$ of size $L\times S$ for individuals in $\mathcal{L}$, a preference ranking matrix $R^\mathcal{S}$ of size $L\times S$ for individuals in $\mathcal{S}$ (where the entries in the last column in $R^\mathcal{L}$, $R^\mathcal{S}$ correspond to the virtual channel), social constraints $C_{\ell}$ for each $\ell\in\mathcal{L}$, and the iteration number $T$. The number of iterations is at most $L \cdot S$, and it is smaller in certain special cases, as will be demonstrated later.

At each iteration, RP\&R cycles through each $s \in \ES$, conducting a Re-Propose step in which $s$ evaluates preferences from most to least favored. If a candidate $\ell$ is accessible and equally or more preferred than the current pairing, $\phi(\ell)$, then $\ell$ is re-assigned to $s$. If $\ell$ and $s$ were paired but a more preferred $\ell'$ who is not socially compatible with $\ell$ becomes available, then the match between $s$ and $\ell$ is dissolved. This ongoing process of evaluation and adjustment enables each proposer $s$ to continuously refine their matches, enhancing the system's overall stability and satisfaction.

Since a stable solution to the SPP with preference ranking is not guaranteed (Theorem \ref{th:Not}), the following subsections will explore several notable communication scenarios where the proposed algorithm does guarantee a solution.

\begin{algorithm}
\caption{Re-Propose \& Reject (RP\&R)}
\begin{algorithmic}[1]
\Require $\mathcal{L}, \mathcal{S}$ $R^\ELL, R^\ES, \mathcal{C}_{\ell} \; \forall \ell \in \ELL, T$ 
\State \textbf{Initialization:} $\forall \ell \in \ELL: \phi(\ell)=V $

\For{$iter = 1,2,\ldots, T$}
    \For{$s = 1,2,\ldots, S$}
        \For{$\ell = \ell^s_1, \ell^s_2, \ldots, \ell^s_L$}
            \If{$s$ is socially available to $\ell$} \Comment{$s$ proposes}
                \If{$R^\ELL_{\ell,s} \leq R^\ELL_{\ell,\phi(\ell)}$} \Comment{$\ell$ accepts}
                \State $phi(\ell)=s$ 
                \EndIf
            \ElsIf{$phi(\ell)=s$}
            \State $phi(\ell)=V$
            \EndIf
        \EndFor
    \EndFor
\EndFor
\end{algorithmic}
\label{alg:rpr}
\end{algorithm}

\subsection{SPP in an Inherently Harmonious System} 

In this scenario, we consider SPP with an empty social constraints graph, ensuring that any solution is harmonious. An illustration of such a graph with $L=6$ is provided below:
\vspace{-0.5cm} 
\begin{figure}[H]
  \centering
  \includegraphics[width=0.2\textwidth]{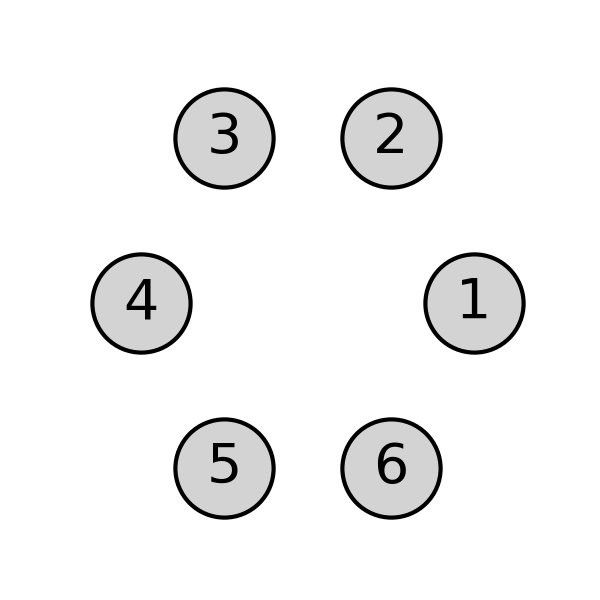}
  \caption{A sample visualization of an empty social constraints graph with 6 nodes ($L=6$).}
  \label{fig:empty_graph}
\end{figure}

An empty (edgeless) graph, meaning there are no social conflicts, and every individual in group $\ELL$ is compatible with everyone else. This situation represents a scenario with very low interference between cells in communication networks. Under these circumstances, the only stable solution is for each cell $\ell \in \ELL$ to be matched with their most preferred channel $s \in \ES$.
\begin{theorem}
    The RP\&R algorithm reaches a stable solution for an empty social constrains graph within 1 external iteration ($T=1$). 
\end{theorem}
\begin{proof}
    We prove the theorem by demonstrating that RP\&R achieves a stable solution after only one Re-Proposing step by each element in $\ES$ (i.e., when $T=1$). Note that when any $s \in \ES$ undergoes Re-Proposing, every $\ell \in \ELL$, for which $s$ ranks highest in its preference list (indicating $\ell$'s utmost preference for $s$ over all other elements in $\ES$), will be paired with $s$. Furthermore, this pairing will remain intact as there are no social constraints compelling $s$ to relinquish $\ell$. Consequently, after a single Re-Proposing cycle for every $s \in \ES$, each $\ell$ will have paired with its top-choice $s$, a match that will persist until the algorithm concludes. Thus, by the end of the algorithm's execution, every $\ell$ will be matched with the $s$ it prefers the most, rendering the solution evidently stable.
\end{proof}

\subsection{SPP in an Inherently Conflicting System}

In this scenario, we examine SPP with a complete social constraints graph, where the graph includes all possible social constraints. The following illustration depicts this setup for $L=6$:

\vspace{-0.5cm} 
\begin{figure}[H]
  \centering
  \includegraphics[width=0.2\textwidth]{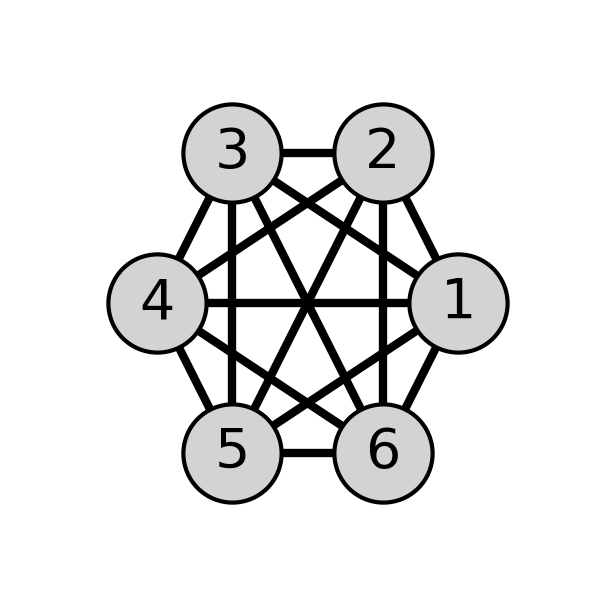}
  \caption{
  A sample visualization of a complete social constraints graph with 6 nodes ($L=6$).}
  \label{fig:complete_graph}
\end{figure}

In this context, complete graph indicates that no two individuals in $\ELL$ are compatible with each other. This scenario represents a situation with extremely high interference in communication networks, where each cell is within the interference range of all other cells.

\begin{theorem}
    The RP\&R algorithm reaches a stable solution for a complete graph within $L$ external iterations ($T=L$). 
\begin{proof}
To prove the theorem, we first show that in the context of a complete graph, if $s$ makes an offer to $\ell$ which is rejected, then $\ell$ will always reject any subsequent proposals from $s$. To establish this, assume, for contradiction, that $\ell$ rejects $s$ during the initial Re-Propose step but later accepts $s$ in a subsequent Re-Propose step.

Consider the earliest such instance, identifying the parties involved as $\ell_1$ and $s_1$. The premise for $\ell_1$'s initial rejection of $s_1$ must be $\ell_1$'s prior engagement with $s_2$, whom it favors over $s_1$. However, for $\ell_1$ to accept $s_1$ in a later iteration, it necessitates that $s_2$ must have dismissed $\ell_1$. This dismissal by $s_2$ could only arise if, during its Re-Propose step, $\ell_2$-who is preferred by $s_2$ over $\ell_1$-accepts $s_2$, whereas previously, since it is a complete graph, $\ell_2$ must have rejected $s_2$. This creates a paradox, suggesting that the situation cannot be the earliest instance of such an event, thereby proving that if $s$ makes an offer to $\ell$ which is turned down, then $\ell$ will always reject any subsequent proposals from $s$. We observe that if the constraint graph were not complete, the claim might not hold true. This is because the reason $\ell_2$ was not previously paired with $s_2$ may not necessarily be due to $\ell_2$'s refusal. it might also be that $s_2$ did not extend an offer to $\ell_2$ initially due to social constraints.

Next, we proceed to prove the theorem. Note that, in scenarios involving a complete graph, where no two individuals from $\ELL$ can coexist peacefully alongside the same individual from $\ES$, the problem essentially reduces to the classical SMP, except for differences in group sizes. Therefore, we apply stability results from SMP for unequal sets, as indicated in \cite{mcvitie1970stable}. The difference in applying the RP\&R algorithm with $T=L$ in this scenario, compared to the method suggested in \cite{mcvitie1970stable}, is that each individual $s$, during the Re-Propose step process, reviews their entire preference list from the start rather than resuming from their most recent selection. Despite this difference, the impact on the final outcome remains unchanged. This is because, based on the claim we proved above, if $s$ extends an offer to $\ell$ and is declined, $\ell$ will not accept any future proposals from $s$, ensuring that both algorithms ultimately facilitate identical pairings. Consequently, a stable solution exists, and RP\&R is guaranteed to reach it.
\end{proof}
\end{theorem}

\subsection{SPP in a Disjoint Union of Complete Social Constraints Graphs}

Building on the results of the previous sections, we now consider a disjoint union of complete social constraints graphs. To prove that RP\&R solves the SPP in this context, we start by establishing the following lemma:

\begin{lemma}
     Each individual $\ell \in \ELL$ can only influence the algorithmic process and future pairings of other individuals within $\ELL$ to which it has a direct or indirect connection via the social constraints graph.
\end{lemma}
\begin{proof}
    Assume, for contradiction, that the theorem is false. This implies there exists an instance $I_1$ of the SPP in which there are individuals $\ell_1$ and $\ell_2$ who belong to different connected components in the graph of social constraints, yet $\ell_1$ affects the algorithmic process of $\ell_2$. In other words, consider another instance $I_2$ of the problem, which is identical to $I_1$ in terms of preference matrices and the graph of social constraints, except that $\ell_1$ does not exist in $I_2$. According to the assumption, the sequence of individuals that $\ell_2$ matches with during the algorithm in $I_1$ is different from the sequence in $I_2$. Let us denote the individuals in $\ES$ (in $I_1$, and $I_2$) in the first differing match in these sequences as $s_1$ and $s_2$.
    The difference in matches arises because, without loss of generality, $s_2$ matches with $\ell_3$, which is socially incompatible with $\ell_2$ in $I_2$, whereas in $I_1$, $s_2$ does not match to $\ell_3$. Now, we need to recursively examine why this difference occurs by tracing the sequence of matches within the connected component of $\ell_2$ until eventually, we reach a point where there is a difference in matching for an individual $\ell''$ within the connected component of $\ell_2$. This implies $\ell''$ is the first individual in the connected component of $\ell_2$ where the matches differ between $I_1$ and $I_2$. However, since the graph of social constraints is identical in this connected component for both instances (and the only difference is in the component of $\ell_1$), there cannot be an additional individual $\ell$ causing this discrepancy. This leads to a contradiction of our initial assumption, which completes the proof.
\end{proof}

Building on the lemma, we now establish the following theorem:

\begin{theorem}
The RP\&R algorithm reaches a stable solution for a disjoint union of complete social constraints graphs within $L_{max}$ external iterations ($T=L_{max}$), where $L_{max}$ represents the largest number of nodes in any of the complete graphs.
    \begin{proof}
        Let $n$ be the number of complete subgraphs, and $\ELL_1, \ELL_2, ..., \ELL_n$ denote the subsets of $\ELL$ corresponding to these subgraphs. Since there are no paths connecting individuals between subgraphs, RP\&R processes these subsets independently. This separation ensures that the algorithm's operations on one subset do not impact the others, allowing for parallel and independent execution across all subgraphs. Given that RP\&R reaches stable solutions within each complete social constraints graph, and by treating these subsets as isolated cases based on the lemma, RP\&R achieves stability across the entire set $\ELL$. Since $L_i$ external iterations are required to achieve stability for each subgraph $\ELL_i$, RP\&R reaches a stable solution within $L_{max}$ external iterations.
    \end{proof}
\end{theorem}

\subsection{SPP in an Acyclic Social Constraints Graphs}

Finally, we highlight an intriguing case of acyclic social constraints graphs. We conducted extensive simulations, exceeding one million random trials, to evaluate the algorithm's performance on acyclic graphs. In all simulations, the algorithm successfully achieved stable polygamy within $T=L$ iterations. Proving this case analytically remains an open problem.

\subsection{Centralized Implementation of RP\&R for Spectrum Access with Channel Reuse}

In the preference ranking model, each cell $\ell \in \mathcal{L}$ can rank the channels they wish to access based on the utility (e.g., rate) they achieve on each channel. Meanwhile, a control manager ranks the preferred users on each channel based on factors such as user characteristics and the potential interference a cell might cause to external usage of the channel or channel operations. Implementing RP\&R to solve SPP with preference ranking for spectrum access with channel reuse can be done by the controller that runs the algorithm, and inform the cells the resulting matches. This mechanism offers centralized implementations with low complexity, making it suitable for integration into 5G and beyond technologies with centralized node deployments.

\section{Simulation Results}
\label{sec:case}

In this section, we compare the proposed DSSAR and RP\&R algorithms under common utility and preference ranking settings, respectively, with well-known benchmarks: (i) \emph{Random Matching:} This algorithm iteratively selects random pairs $(s, \ell)$ from all possible combinations, updating the potential pairs based on $\ell$'s  social constraints. Specifically, it removes from the potential pairs any pair that includes the selected $s$ and $\ell'$ that is socially incompatible with $\ell$. (ii) \emph{Best of Random:} This algorithm runs the Random Matching process multiple times (in our setting, $L$ times) and selects the match that yields the highest social welfare. In the context of the common utility setting, Best of Random serves only as a benchmark since, unlike DSSAR, it cannot be executed in a distributed manner. (iii) \emph{Top-Ranked Proposal:} In this algorithm, each $\ell$ proposes to its most preferred $s$. Each $s$ then accepts proposals based on its preference order, adhering to social constraints by eliminating offers from individuals in $\ELL$ who are socially incompatible with those it has already matched with. (iv) \emph{Optimal Polygamy:} For small systems where exhaustive search is feasible within a reasonable timeframe, we computed the optimal polygamy by implementing an exhaustive search. Specifically, we conducted an exhaustive search over all $L^{(S+1)}$ possible matches, accounting for all potential pairings between $\mathcal{L}$ and $\mathcal{S}$, including cases where $\ell$ remains unmatched. For each matching, we checked if it satisfied the criteria for harmonious polygamy and calculated the corresponding social welfare. Due to the exponential complexity of this process, we limited our experiments to cases where $L$ ranged from $3$ to $9$ and $S$ ranged from $2$ to $3$, with random selection for each experiment.

We evaluated the algorithms under both common utility and preference ranking settings. In the common utility setting, a cell-channel $(\ell, s)$ match was assessed based on the data rate computed using Shannon capacity. For the preference ranking setting, three measures were computed: First, $\mathcal{S}$'s social welfare, where each match between $s$ and $\ell$ increased welfare according to $\ell$'s ranking in $s$'s preferences (e.g., if $\ell$ was $s$'s top choice, welfare increased by $L$; if second, by $L-1$, and so on); second, $\mathcal{L}$'s social welfare, defined similarly but based on $s$'s ranking in $\ell$'s preferences; and finally, the total welfare, which was the average of $\mathcal{S}$'s and $\mathcal{L}$'s welfare. Each individual's preferences are randomly selected from all possible permutations, and the social constraint graph is modeled as a random geometric graph.

We conducted $10,000$ experiments, with each experiment generating a social constraints graph, a utility matrix, and preference rankings as previously described. The results of these experiments are presented in Table \ref{tab: small L} and Figures \ref{fig:util, small L} and \ref{fig:pref, small L}. On average, the RP\&R algorithm achieved $96\%$ of the optimal social welfare in the preference ranking simulations and $97\%$ in the common utility simulations. These results are comparable to the performance observed in the special case of SMP for spectrum access \cite{leshem2011multichannel}. Notably, RP\&R significantly outperforms all other methods.

\begin{table}[H]
    \centering
    \resizebox{\columnwidth}{!}{ 
    \begin{tabular}{lccc}
        \toprule
        & \makecell{\textbf{Total Welfare}}  
        & \makecell{\textbf{$S$'s social Welfare}} 
        & \makecell{\textbf{$L$'s social welfare}}\\
        \midrule
        Random Matching & 0.472 & 0.436 & 0.508 \\
        Top-Ranked Proposal & 0.511 & 0.423 & 0.598 \\
         Best-of-Random & 0.554 & 0.514 & 0.594 \\
        \textbf{RP\&R (ours)} & \textbf{0.582} & \textbf{0.557} & \textbf{0.607}\\
        \midrule
         Optimal Polygamy & 0.606 & 0.561 & 0.651 \\ 
        \bottomrule
    \end{tabular}
    }
    \caption{The results of preference ranking for a small network across different algorithms, averaged over 10,000 random experiments. In each experiment, $L$ was randomly selected between $3$ and $9$, and $S$ was chosen between $2$ and $3$.}
    \label{tab: small L}
\end{table}

\begin{table}[H]
    \centering
    \resizebox{\columnwidth}{!}{ 
    \begin{tabular}{lccc}
        \toprule
        & \makecell{\textbf{Total Welfare}}  
        & \makecell{\textbf{$S$'s social Welfare}} 
        & \makecell{\textbf{$L$'s social welfare}}\\
        \midrule
        Random Matching & 0.321 & 0.301 & 0.342 \\
        Top-Ranked Proposal & 0.358 & 0.272 & 0.445 \\
        Best-of-Random & 0.371 & 0.347 & 0.395 \\
        \textbf{Proposed RP\&R (ours)} & \textbf{0.445} & \textbf{0.522} & \textbf{0.478}\\
        \bottomrule
    \end{tabular}
    }
    \caption{The results of preference ranking for a larger network across different algorithms, averaged over 10,000 random experiments. In each experiment, $L$ was randomly selected between $20$ and $100$, and $S$ was chosen between $2$ and $10$.}
    \label{tab: big L}
\end{table}

\begin{figure}[!ht]
    \centering
    \includegraphics[width=0.9\linewidth]{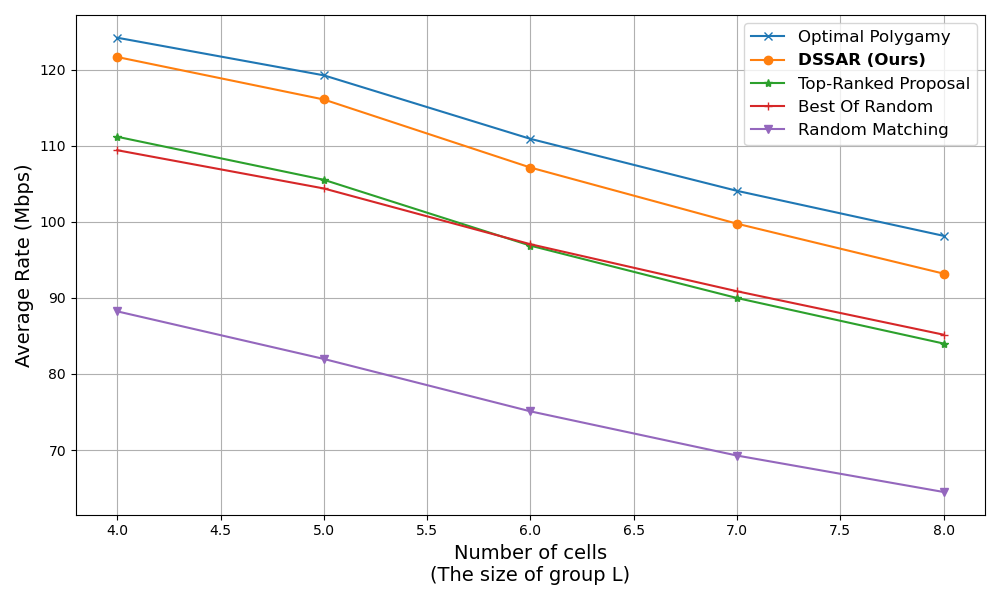}
    \caption{The results for common utility for a small network, measured by data rate as a function of the number of cells ($L$), averaged across 10,000 random experiments for various algorithms.}
    \label{fig:util, small L}
\end{figure}

\begin{figure}[!ht]
    \centering
    \includegraphics[width=0.9\linewidth]{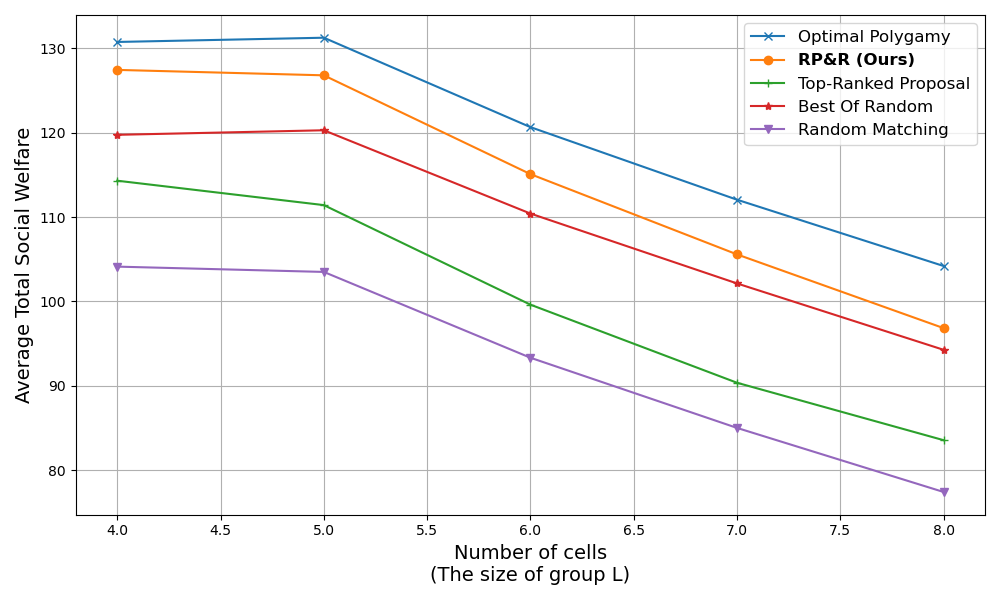}
    \caption{The results for total social welfare for a small network as a function of the number of cells ($L$), averaged across 10,000 random experiments for various algorithms.}
    \label{fig:pref, small L}
\end{figure}

\begin{figure}[h!]
    \centering
    \includegraphics[width=0.9\linewidth]{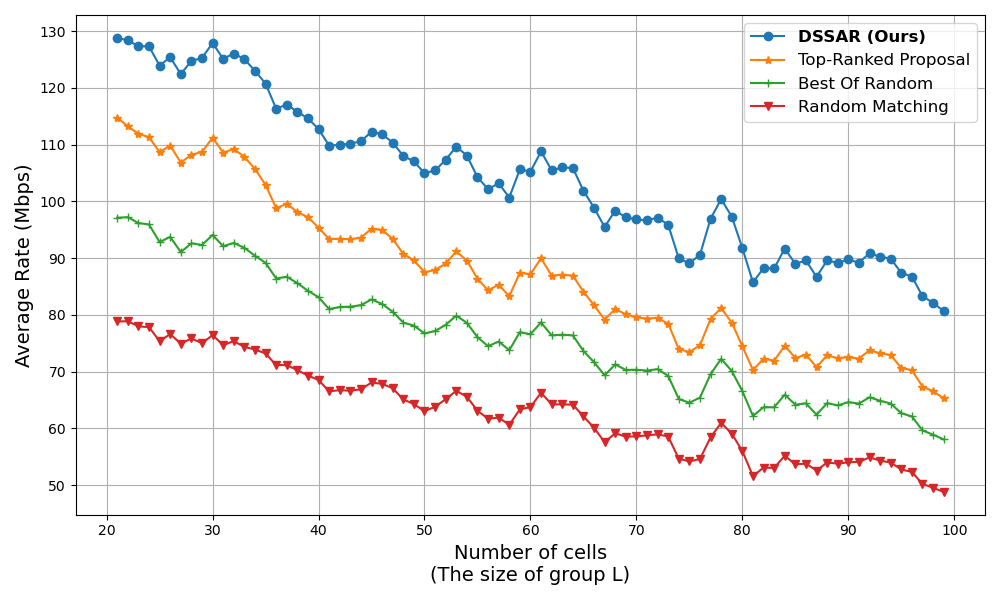}
    \caption{The results for common utility for a larger network, measured by data rate as a function of the number of cells ($L$), averaged across 10,000 random experiments for various algorithms.}
    \label{fig:util, big L}
\end{figure}

\begin{figure}[h!]
    \centering
    \includegraphics[width=0.9\linewidth]{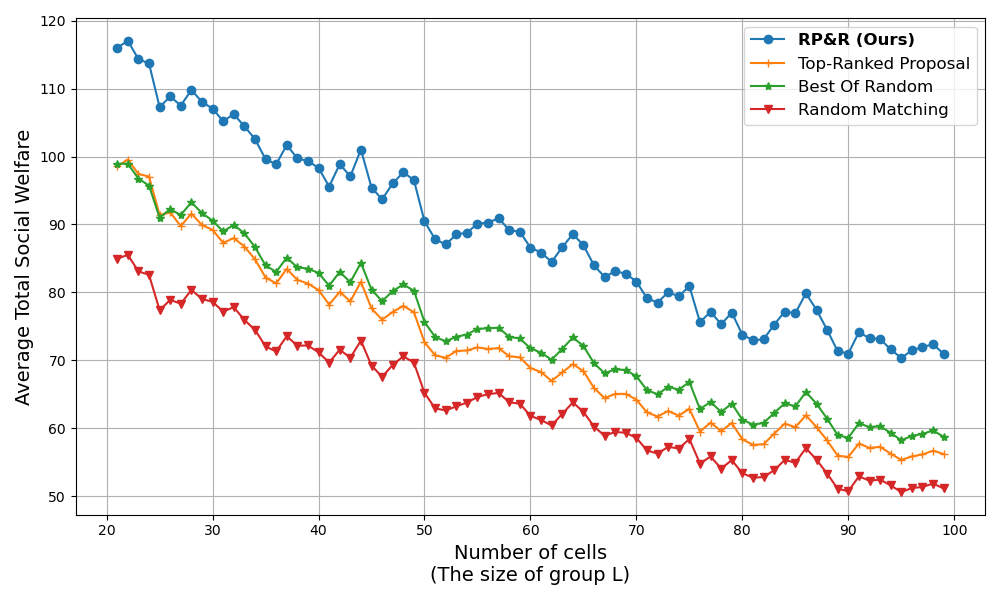}
    \caption{The results for total social welfare for a larger network as a function of the number of cells ($L$), averaged across 10,000 random experiments for various algorithms.}
    \label{fig:pref, big L}
\end{figure}

\begin{figure}[h!]
    \centering
    \includegraphics[width=0.9\linewidth]{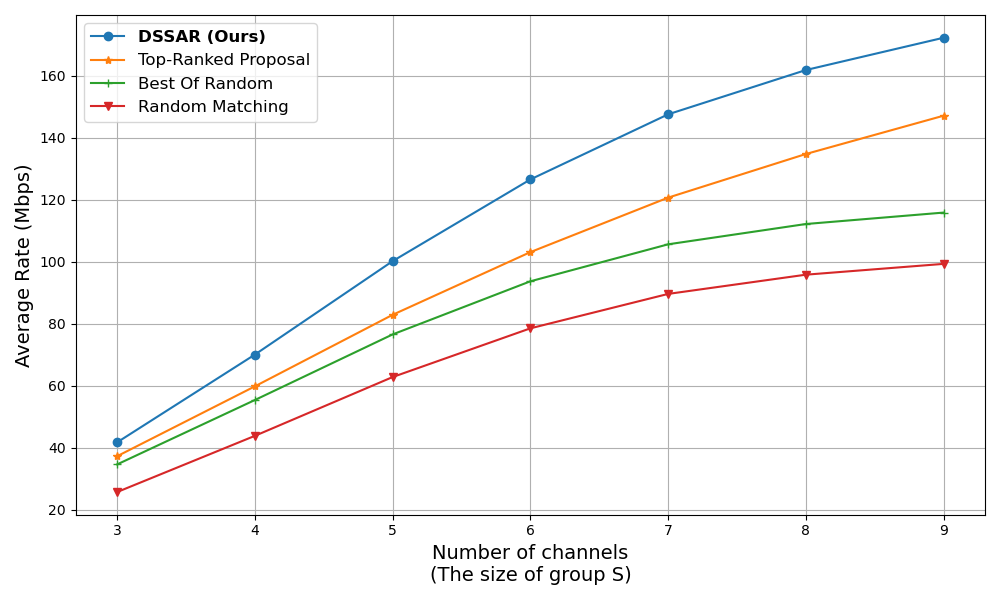}
    \caption{The results for common utility for a larger network, measured by data rate as a function of the number of channels ($S$), averaged across 10,000 random experiments for various algorithms.}
    \label{fig:util, big S}
\end{figure}

\begin{figure}[h!]
    \centering
    \includegraphics[width=0.9\linewidth]{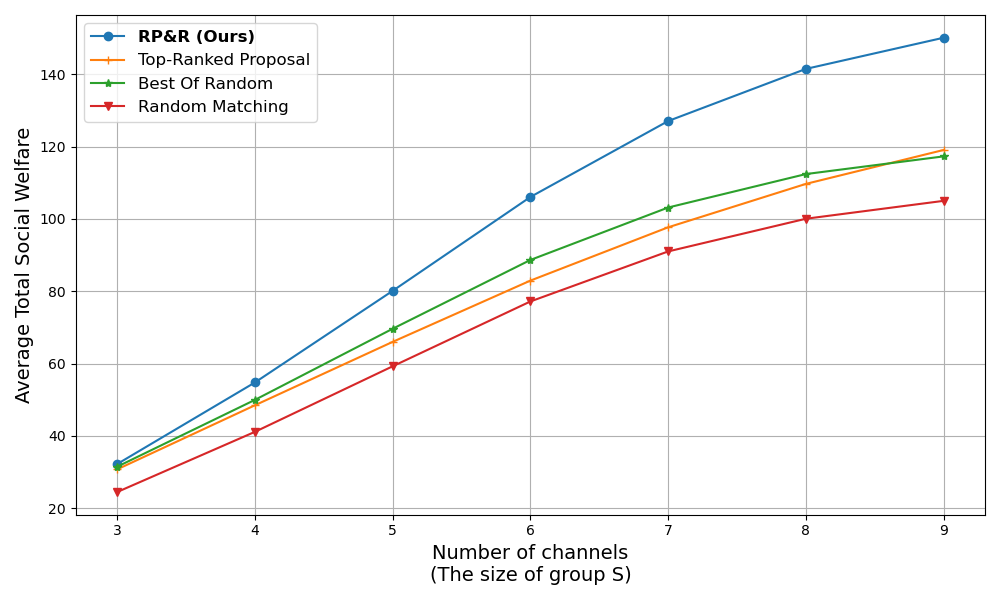}
    \caption{The results for total social welfare for a larger network as a function of the number of channels ($S$), averaged across 10,000 random experiments for various algorithms.}
    \label{fig:pref, big S}
\end{figure}

We also evaluated the scalability of our method compared to other methods by testing larger numbers of channels ($S$) and cells ($L$), as detailed in Table \ref{tab: big L} and Figures \ref{fig:util, big L}-\ref{fig:pref, big S}. For this, we conducted an additional $10,000$ experiments using the same parameters as before, but with $S$ randomly varying between $2$ and $9$ and $L$ ranging from $20$ to $100$. It is important to note that for such large values of $L$, performing an exhaustive search for the optimal polygamy becomes impractical due to the associated complexity. Notably, DSSAR in the common utility setting and RP\&R in the preference ranking setting both significantly outperform all other methods.

\section{Conclusion}
\label{sec:conclusion}

In this paper, we introduced and analyzed the stable polygamy problem (SPP), a new and broader formulation of the stable marriage problem (SMP). Unlike the classic SMP, which only allows one-to-one matching, the SPP accommodated multiple individuals from a larger group being matched with a single individual from a smaller group, while incorporating social constraints that dictate which individuals cannot coexist harmoniously. We demonstrated that traditional algorithms like propose and reject and the Hungarian method are inefficient in this new setting. To address this, we developed efficient algorithms that solve the SPP in polynomial time for spectrum access with channel reuse. Our analysis and simulations confirmed the effectiveness of these algorithms in various spectrum access scenarios.

\bibliographystyle{ieeetr}

\end{document}